\documentclass[envcountsame]{llncs}
\title{Attacks on the Search-RLWE problem with small errors}
\author{Hao Chen\inst{1} \and Kristin Lauter\inst{2} \and Katherine E. Stange \inst{3}}

\institute{Microsoft Research, Redmond, USA \\
 \email{haoche@microsoft.com} \\
 \and
Microsoft Research, Redmond, USA \\
 \email{klauter@microsoft.com} \\
\and
University of Colorado, Boulder, USA \\
\email{kstange@math.colorado.edu}}
\usepackage{macros}
\usepackage{stackrel}
\usepackage{float}
\usepackage{url}
\usepackage{placeins}
\usepackage{hyperref}
\spnewtheorem{fact}{Fact}{\bfseries}{\rmfamily}

\usepackage{amsmath}
\renewcommand{\Mod}[1]{\ (\mathrm{mod}\ #1)}

\usepackage{algpseudocode}
\usepackage{algorithmicx}
\usepackage{algorithm}
\usepackage{physics}

\DeclareMathOperator{\disc}{disc}
\usepackage[textwidth=50,textsize=tiny]{todonotes}
\begin{document}
\maketitle

\begin{abstract}
  The Ring Learning-With-Errors (RLWE) problem shows great promise for post-quantum cryptography and homomorphic encryption.  We describe a new attack on the non-dual search RLWE problem with small error widths, using ring homomorphisms to finite fields and the chi-squared statistical test. 
In particular, we identify a ``subfield vulnerability'' (Section 5.2) and give a new attack which finds this vulnerability
by mapping to a finite field extension and detecting non-uniformity with respect to
the number of elements in the subfield.
We use this attack to give examples of vulnerable RLWE instances in Galois number fields. We also extend the well-known search-to-decision reduction result to Galois fields with any unramified prime modulus $q$, regardless of the residue degree $f$ of $q$, and we use this in our attacks. The time complexity of our attack is $O(n q^{2f})$, where $n$ is the degree of $K$ and $f$ is the {\it residue degree} of $q$ in $K$. We also show an attack on the non-dual (resp. dual) RLWE problem with narrow error distributions in prime cyclotomic rings when the modulus is a ramified prime (resp. any integer). We demonstrate the attacks in practice by finding many vulnerable instances and successfully attacking them.  We include the code for all attacks. \\ \\
% Then we analyze the security of cyclotomic fields against our attack.
%Also, we sharpen the attack in \cite{elos2015weak} and give examples of vulnerable instances of cryptographic size. Finally, we discuss the effect of modulus switching on our attacks.
\end{abstract}

\section{Introduction}

%\red{For Kristin}: \cite{castryck2016provably}. \cite{peikert2015notes}. \cite{cryptoeprint:2015:971}.

The Ring Learning-with-Errors (RLWE) problem, proposed  in  \cite{lyubashevsky2013ideal}, is a variant of the Learning-with-Errors (LWE) problem, and is an active research area in lattice based cryptography, and a candidate for post-quantum cryptography. It has drawn increased attention because it can be used for homomorphic encryption (\cite{bos2013improved,brakerski2012leveled,brakerski2011fully,brakerski2014efficient,gentry2012fully,lopez2012fly,stehle2011making}).  The problem, which comes in \emph{search} and \emph{decision} variants, is based on the geometry of ideal lattices, in particular the rings of integers of number fields, or their duals.

It is of great importance to understand the security of RLWE.  The first piece of the puzzle is provided by proofs of security \cite{lyubashevsky2013ideal}.  However, it is also important to mount direct attacks on the problem and its variants, in order to illuminate the protective properties of the provably secure instances, the dangers of deviating from established parameters, and the practical runtimes for certain parameter sizes.  
This work is part of that programme, which has seen much recent interest, e.g.   \cite{eisentrager2014weak,elos2015weak}.  An eprint version of the current paper \cite{chenattacks} has already generated much follow-up work \cite{castryck2016error,castryck2016provably,chen2016vulnerable,peikert2016not}.  In this paper we provide a brief overview of past work, and then present attacks which are novel in their mathematical underpinnings (based on new homomorphisms to finite fields of higher degree which detect in particular ``subfield vulnerabilities'', see Section~\ref{subfield vulnerability}).
We also discuss the underlying number theory and geometry of these attacks to provide a framework for future work.

An instance of the RLWE problem is determined by a choice of a number field $K$ and a prime $q$ called the {\it modulus}, along with an error distribution. The authors of \cite{lyubashevsky2013ideal} proved a reduction from certain hard lattice problems to an instance of search RLWE involving a continuous Gaussian error distribution modulo the dual ideal $R^{\vee}$ of the ring of integers $R$ of $K$. Ducas and Durmus proposed a non-dual variant of RLWE in the cyclotomic setting and proved its hardness in \cite{ducas2012ring}. Also in \cite{lyubashevsky2013ideal}, a search-to-decision reduction was proved for RLWE problems in cyclotomic fields and modulus $q$ which splits completely. This reduction was then generalized in \cite{eisentrager2014weak} to hold for general Galois number fields where $q$ splits.  As an auxiliary result in this paper, we generalize this search-to-decision reduction to work for the case of unramified modulus $q$ of arbitrary degree.

The non-dual variant of RLWE generates the error distribution as a discrete Gaussian on the ring of integers $R$ under the canonical embedding, instead of in the image of the dual ideal.  The dual and non-dual variants are equivalent up to a change in the error distribution (see Section \ref{sec: background}).  For the non-dual variant of RLWE, the authors of \cite{elos2015weak} proposed an attack on the {\it decision} RLWE problem. The attack makes use of a modulus $q$ of residue degree $1$, giving a ring homomorphism $\rho: R \to \bF_q$ (so that it could be called a \emph{mod $q$} attack, although it differs from \cite{micciancio2009lattice,nguyen1999cryptanalysis}).  The attack works when, with overwhelming probability, the image of the RLWE error distribution under the map $\rho$ takes values only in a small subset of $\bF_q$.  The authors of \cite{elos2015weak} then gave an infinite family of examples vulnerable to the attack. Unfortunately, the vulnerable number fields in \cite{elos2015weak} are not Galois. Hence, the search-to-decision reduction theorem does not apply, and the attack can not be directly used to solve the search variant of RLWE for those instances.

In this paper, we generalize the attack of \cite{elos2015weak} to certain Galois number fields and moduli of higher degree. As a result, we have an attack on the {\it search} RLWE problem and an implementation of the attack on concrete RLWE instances, including the search-to-decision reduction.  Our attack is new in two major ways: first, the attack considers ring homomorphisms from $R \to \bF_{q^f}$, for $f>1$, instead of just homomorphisms from $R \to \bF_q$ (so it is no longer `mod $q$'); second, the error distribution is distinguished from random (i.e. from the uniform distribution) using the statistical chi-squared test, instead of relying on the values of the error polynomial to be small or in a small subset.  The attack aims at an intermediate problem used in the search-to-decision proof of \cite{lyubashevsky2013ideal}, which is to recover the secret modulo a prime ideal (denoted $\SRLWE(\cR,\fq)$; see Definition~\ref{def: srlwe mod q}). The time complexity of our attack is $O(nq^{2f})$, where $n$ is the degree of $K$ and $f$ is the {residue degree} of $q$ in $K$.

Importantly, we also show an attack on prime cyclotomic rings under certain assumptions on the modulus and error rate, which succeeds with high probability and with surprising efficiency.  First, we give attacks for the {\it decision} version of the non-dual variant of RLWE considered in \cite{elos2015weak}, when the modulus $q$ is equal to the unique ramified prime $p$. For example, we show that in dimension $n=808$, we can attack an RLWE instance in the cyclotomic ring $\mathbb{Q}(\zeta_{809})$ effectively in $35$ seconds, where the modulus is $809$.  This opens up the question of whether general cyclotomic fields are safe for cryptography, depending on whether modulus switching can be used to transfer this attack from the ramified modulus to other larger moduli which are used in practice.  In addition, we attack the {\it decision}
version of the {\it dual} RLWE problem on the $p$-th cyclotomic field with {\it arbitrary} modulus $q$, assuming that the width $r$ of the error distribution is around $\frac{1}{\sqrt{p}}$.

The error widths for which our attacks work are below those required by the security proof of \cite{lyubashevsky2013ideal}, which requires $r = \omega(\sqrt{\log n})$. In particular, this work does not affect the hardness results of \cite{lyubashevsky2013ideal}.  On the other hand, in practice, in implementations of homomorphic encryption systems based on the hardness of RLWE, it has been common practice to use small errors to improve efficiency for the systems. We show in this work that for errors in the width range below provably secure but above linear algebra vulnerable (errorless LWE), the security of RLWE depends in an interesting way on the choice of ring and modulus.  To be more specific, the geometry of the ring of integers and the manner in which certain prime ideals exist as sublattices are important factors (see Section 5.2). Finally, it is important to note that most implementations of RLWE-based schemes use exclusively 2-power cyclotomic rings, on which our attacks are not effective. Hence the impact of our attacks on the security of existing practical implementations of RLWE-based homomorphic encryption schemes is limited.

Auxiliary results we present include several stand-alone items of possibly independent interest: we prove a search-to-decision reduction  for Galois fields which applies for any unramified modulus $q$, regardless of the residue degree of $q$ (this relies heavily on Galois theory and Galois fields are the largest class to which we expect this to apply).  We also present some heuristic arguments as to whether modulus switching techniques are likely to be successfully combined with our attacks.

We end this section with a table summarizing what is known about the security of RLWE for certain choices of number fields.  The first table deals with the continuous dual version. For comparison, we normalize the error width:  let $\tilde{r} = r \cdot |d_K|^{1/2n}$, where $d_K$ is the discriminant of the number field $K$. 

\begin{table}[H]
\begin{center}
\caption{Security of dual RLWE}
\begin{tabular}{c|c|c|c}
field & modulus & $\tilde{r}$ & security \\
\hline
$\bQ(\zeta_{m})$ & {$q \equiv1 \mod{m}$, $q = poly(m)$} &  $\omega(\sqrt{\log n}) \cdot \Theta(\sqrt{n})$ & decision is secure \cite{lyubashevsky2013ideal} \\ \hline
Any & $q = poly(n)$ & $\omega(\sqrt{\log n}) \cdot |d_K|^{1/2n}$ & search is secure \cite{lyubashevsky2013ideal} \\ \hline
$\bQ(\zeta_p)$  & any & $\sim 1$ &  decision is not secure (this paper)\\
\end{tabular}
\end{center}
\end{table}

The second table deals with the non-dual discrete version.
Here we normalize by $\tilde{r} = r/|d_K|^{1/2n}$.
The heuristic expectation is that when $\tilde{r} =\Omega(\sqrt{n})$ and $q = poly(n)$, decision RLWE problem should be hard.

\begin{table}[H]
\begin{center}
\caption{Security of non-dual RLWE}
\begin{tabular}{c|c|c|c}
field & modulus & $\tilde{r}$ & security \\
\hline
%$\bQ(\zeta_{2^d})$ & {$q \equiv1 \mod{2^d}$, $q = 2^{O(n)}$} & $\sim 3$ & \red{not} search \cite{laine2015key} \\ \hline

%$\bQ(\zeta_{m})$ & {$q \equiv1 \mod{m}$  $q = poly(m)$} &  $\omega(\sqrt{\log n})$ & (likely) decision \cite{chen2016vulnerable} \\ \hline
$\bQ(\sqrt[n]{1-q})$ & $poly(n)$  & $\sim 1$ & decision is not secure \cite{elos2015weak} \\ \hline
$\bQ(\sqrt[n]{1-q})$ & $poly(n)$  & $\sim 1$ & search is not secure \cite{castryck2016provably} \\ \hline
certain $\bQ(\zeta_m)^H$ & {$poly(n)$ w. properties} & $\sim 1$ & search is not secure (this paper) \\ \hline
$\bQ(\zeta_p)$  & $p$ & $\sim 1$ &  decision is not secure (this paper) \\ \hline
certain $\bQ(\zeta_p, \sqrt{d})$ & {$poly(n)$  w. properties} & $o(\sqrt{d/p})$ &  search is not secure \cite{chen2016vulnerable} \\
\end{tabular}
\end{center}
\end{table}

\subsection{Organization}

In Section \ref{sec: background}, we review some definitions related to the RLWE problems. In Section \ref{sec: s-to-d}, we prove a search-to-decision reduction for Galois extensions $K$ and unramified moduli. In Section \ref{sec: chi-square}, we introduce an attack on non-dual RLWE problems based on the chi-square statistical test, which directly generalizes the attack in \cite{elos2015weak}.
In Section \ref{sec: sub-cyclotomics}, we give examples of subfields of cyclotomic fields vulnerable to our new attack, where the modulus $q$ has residue degree two.
In Section \ref{sec: ramified-prime}, we give attacks on the non-dual RLWE in prime cyclotomic fields when the modulus is the unique ramified prime and dual RLWE for any modulus, assuming the errors are sufficiently narrow.  In Section \ref{sec:mod}, we consider the possibility of combining modulus switching with our attack.

%Finally, in Section \ref{sec: cyclo-secure}, we use Fourier analysis to give a heuristic argument, which we then combine with numerical evidence to support the view that cyclotomic extensions with unramified moduli of small residue degree are invulnerable to our attack.

All computations in this paper were performed in Sage \cite{sage}. All the relevant code is available and  can be found at \url{https://github.com/haochenuw/GaloisRLWE}.

\subsection*{Acknowledgements}
The authors thank Chris Peikert for helpful discussions, and the anonymous referees for helpful suggestions.
The third author was supported by NSF grant DMS-1643552.

\section{Background} \label{sec: background}

Let $K$ be a number field of degree $n$ with ring of integers $R$, and let $\sigma_1, \cdots, \sigma_n$ be the embeddings of $K$ into the field of complex numbers. We define the {\it adjusted canonical embedding} of $K$ as follows: Let $r_1$, $r_2$ denote the number of real embeddings and conjugate pairs of complex embeddings of $K$. Without loss of generality, assume $\sigma_1, \cdots, \sigma_{r_1}$ are the real embeddings of $K$ and $\sigma_{r_1+r_2+j} = \overline{\sigma_{r_1 + j}}$ for $1 \leq j \leq r_2$. Then the adjusted canonical embedding is the
following map:

  \begin{align}
   \iota: K  \to   \bR^n : \quad x \mapsto  & \begin{bmatrix}
           \sigma_1(x) \\
           \vdots \\
           \sigma_{r_1}(x) \\
         \sqrt{2}  \Re(\sigma_{r_1+1})(x) \\
          \sqrt{2} \Im(\sigma_{r_1+1})(x) \\
          \vdots \\
         \sqrt{2} \Re(\sigma_{r_1+r_2})(x) \\
         \sqrt{2} \Im(\sigma_{r_1+r_2})(x).
         \end{bmatrix}
  \end{align}

%\[
 %   \iota: K  \to   \bR^n : \quad x \mapsto (\sigma_1(x), \cdots, \sigma_{r_1}(x), \Re(\sigma_{r_1+1})(x), \Im(\sigma_{r_1+1})(x), \cdots,  \Re(\sigma_{r_1+r_2})(x), \Im(\sigma_{r_1+r_2})(x)).
%\]

It turns out that $\Lambda_R := \iota(R)$ is a lattice in $\bR^n$. Let $w = (w_1, \cdots , w_n)$ be an integral basis for $R$. The embedding matrix of $w$, denoted by $A_w$, is the $n$-by-$n$ matrix whose $i$-th column is $\iota(w_i)$. Note that the columns of $A_{w}$ form a basis for the lattice $\Lambda_R$.

%The two embedding matrices are related in a simple way: let $T$ be the unitary matrix
%\[
%T = \begin{bmatrix}
%    I_{r_1}  & 0  \\
%    0     & T_{r_2} \\
%\end{bmatrix},
%\mbox{ where } T_s = \frac{1}{\sqrt{2}} \begin{bmatrix}
%    I_{r_2}  & I_{r_2} \\
%    -iI_{r_2}     & iI_{r_2} \\
%\end{bmatrix},
%\]
%Then we have
%$$\widetilde{A_{w}} = T A_{w},$$

For $\sigma > 0$, define the Gaussian function $\rho_\sigma: \bR^n \to (0,1]$, depending on the usual inner product on $\mathbb{R}^n$, to be $\rho_\sigma(x) = e^{-||x||^2/2\sigma^2}$.
%(our $\sigma$ is equal to $r/\sqrt{2\pi}$ for the parameter $r$ in \cite{lyubashevsky2013ideal}).
\begin{definition}
For a lattice $\Lambda \subset \bR^n$ and $\sigma > 0$, the {\it discrete Gaussian distribution} on $\Lambda$ with parameter $\sigma$ is:
\[
    D_{\Lambda, \sigma}(x) = \frac{\rho_\sigma(x)}{\sum_{y \in\Lambda} \rho_\sigma(y)}, \, \forall x \in \Lambda.
\]
\end{definition}
Equivalently, the probability of sampling any lattice point $x$ is proportional to $\rho_\sigma(x)$. \\

%\subsection{Ring-LWE Problems for General Number Fields}

We follow \cite{elos2015weak} in setting up the non-dual RLWE problem for general number fields.  In particular, the error distribution we use is a spherical discrete Gaussian distribution on $\Lambda_R$.
\begin{definition}
A (non-dual) {\it RLWE instance} is a tuple $\cR = (K,q,\sigma,s)$, where $K$ is a number field with ring of integers $R$, $q$ is a prime, $\sigma >0$ is a positive real number, and $s \in R/qR$ is called the {\it secret}.
\end{definition}

Suppose $\cR = (K,q,\sigma,s)$ is an RLWE instance and let $R$ be the ring of integers of $K$. The {\it error distribution} of $\cR$ is the discrete Gaussian distribution $D_{\Lambda_R,\sigma}$.  \\

% The notation $x \gets D$ indicates that variable $x$ is distributed according to distribution $D$.

Let $R_q$ denote the quotient ring $R/qR$;  then a (non-dual) RLWE sample is a pair
$$(a, b = as+e) \in R_q \times R_q, $$
where the first coordinate $a$ is chosen uniformly at random in $R_q$, and $e$ is sampled from the error distribution and considered modulo $qR$.   The reader unfamiliar with this problem should consider this analogous to a discrete logarithm pair $(g, g^s) \in \mathbb{F}_q \times \mathbb{F}_q$, where $s$ is a secret exponent.

%We use the shorthand notation $(a,b) \gets \cR$ to represent that $(a,b)$ is sampled from the RLWE distribution of $\cR$.

\begin{definition}[Search RLWE]
  Let $\cR$ be an RLWE instance. The \emph{search RLWE problem}, denoted by $\SRLWE(\cR)$, is to discover $s$ given access to arbitrarily many independent samples $(a,b)$.
\end{definition}

\begin{definition}[Decision RLWE]
Let $\cR$ be an RLWE instance. The \emph{decision RLWE problem}, denoted by $\DRLWE(\cR)$, is to distinguish between the same number of independent samples in two distributions on $R_q \times R_q$. The first is the RLWE distribution of $\cR$, and the second consists of uniformly random and independent samples from $R_q \times R_q$.
\end{definition}

\begin{remark}
As pointed out in \cite{elos2015weak}, when analyzing the error distribution, one needs to take into account the sparsity of the lattice $\Lambda_R$, which is measured by its covolume $\det(\Lambda_R) = \sqrt{|\disc(K)|}$. In light of this, we define a relative version of the standard deviation parameter: $\sigma_0  = \frac{\sigma}{|\disc(K)|^{\frac{1}{2n}}}.$
\end{remark}

\begin{remark}
There are different approximate algorithms to sample from discrete Gaussian distributions on lattices. In this paper, we choose to use the sampling algorithm developed  in \cite{gentry2008trapdoors}.
\end{remark}

We now discuss dual RLWE and its relation to non-dual RLWE.  In dual RLWE, the secret $s$ lies in $R^\vee_q := R^{\vee}/qR^{\vee}$, where $R^{\vee}$ is the dual ideal of $R$, and the error $e$ is sampled from $R^\vee$ with discrete spherical Gaussian distribution with width $r = \sqrt{2 \pi} \sigma$.  Therefore the RLWE samples are of the form
\[
  (a, b = as + e ) \in R_q \times R^\vee_q.
\]

%For dual RLWE, the sample lies in the set $R_q \times \frac{K_\bR}{qR^{\vee}}$, where $K_\bR := \iota(K)\otimes_\bQ \bR \cong \bR^n$.

The dual and non-dual versions of the RLWE problem 
%was initially introduced in the dual form, concerning $R^\vee$ instead of $R$ \cite{lyubashevsky2013ideal}.
are very closely related when the dual ideal is principal: in this case, $R$ and $R^\vee$ are related by a scaling factor (which may alter a spherical Gaussian to an ellipsoidal one).  Even when $R^\vee$ is not principal, we have $R \subseteq k_1 R^\vee$ and $R^\vee \subseteq k_2R$ for some constants $k_1$ and $k_2$, so that a problem in one formulation can be reduced to a problem in the other, with a different error distribution.  Several of the non-dual examples of this paper are known to have principal dual ideal \cite{peikert2016not}.  In particular, our attack on the ramified prime for cyclotomic fields is translated to the dual situation in Section \ref{attackingdual}.  The  full class of elliptic Gaussians (not just spherical) is also considered in the security reductions of \cite{lyubashevsky2013ideal}.

Finally, there is a \emph{continuous} version of RLWE which is more amenable to security reductions.  Since one can always discretize, the continuous version reduces to the discrete one presented here, which is more practical for applications.

\section{Search-to-Decision Reduction}
\label{sec: s-to-d}

In \cite{eisentrager2014weak}, the search-to-decision reduction of \cite{lyubashevsky2013ideal} is extended to RLWE for Galois number fields, where $q$ is an unramified prime of degree one.  The approach is via an intermediate problem, denoted $\fq_i$-LWE in \cite{lyubashevsky2013ideal}.  In this section, we extend this result to primes $\fq$ of arbitrary residue degree.  Our intermediate problem, which we denote by $\SRLWE(\cR,\fq)$, is the same as $\fq_i$-LWE, and it amounts to finding the secret modulo the prime $\fq$.  The Galois group allows us to bootstrap this piece of information to discover the full secret.

The attack in Section~\ref{sec: chi-square} targets $\SRLWE(\cR,\fq)$ and hence, by the results of this section, will solve Search RLWE.  In Section~\ref{sec: sub-cyclotomics}, we demonstrate the attack on Search RLWE in practice.

%The main result of this section (Corollary~\ref{cor: s-to-d}) is a reduction from SRLWE to DRLWE for Galois fields $K$ and unramified primes $q$. We will prove the reduction from SRLWE to an intermediate problem, which we denote by $\SRLWE(\cR,\fq)$ (it is denoted by $\fq_i$-LWE in \cite{lyubashevsky2013ideal}). This result can be viewed as a generalization of \cite[Theorem 2]{eisentrager2014weak} to primes of higher degree. Since our attack in Section~\ref{sec: chisquare} is  targeting $\SRLWE(\cR,\fq)$, we could attack SRLWE for any Galois RLWE instances vulnerable to our attack.

%We remark that a search-to-decision reduction theorem for higher degree primes can be proved by carrying out almost the exact same proof in \cite{eisentrager2014weak}.

\begin{definition} \label{def: srlwe mod q}
        Let $\cR = (K,q,\sigma, s)$ be an RLWE instance and let $\fq$ be a prime of $K$ lying above $q$.  The problem $\SRLWE(\cR, \fq)$ is to determine $s \pmod {\fq}$, given access to arbitrarily many independent samples $(a,b) \gets \cR$.
\end{definition}

We recall some facts from algebraic number theory in the following lemma.
\begin{lemma}
\label{lem: prime factorization}
Let $K/\bQ$ be a finite Galois extension of degree $n$ with ring of integers $R$,  and let $q$ be a prime unramified in $K$. Then there exists a unique divisor $g$ of $n$ and a set of $g$ distinct prime ideals $\fq_1, \cdots ,\fq_g$ of
$R$ such that:
\begin{enumerate}
        \item $qR = \prod_{i=1}^g \fq_i$,
        \item the quotient $R/\fq_{i}$ is a finite field of cardinality $q^f$  for each $i$, where $f = \frac{n}{g}$,
        \item there is a canonical isomorphism of rings
                \begin{equation}
                        \label{eqn: Rq-factor}
    R_q \cong R/\fq_{1} \times \cdots \times R/\fq_{g},
    \end{equation}
\item the Galois group acts transitively on the ideals $\fq_1, \ldots, \fq_g$ and this action descends to an action on $R_q$ which permutes the corresponding factors in \eqref{eqn: Rq-factor} in the same way.
\end{enumerate}
\end{lemma}
The number $f$ in the above lemma is called the {\it residue degree} of $q$ in $K$. Note that the prime $q$ splits completely in $K$ if and only if its residue degree is one.

\begin{theorem} \label{thm: reduction}
        Let $\cR = (K,q,\sigma, s)$ be an RLWE instance such that $K/\bQ$ is Galois of degree $n$ and $q$ is unramified in $K$ with residue degree $f$. Let $\sA$ be an oracle which solves $\SRLWE(\cR,\fq)$ using a list of $m$ samples modulo $\fq$.  Let $S$ be a set of $m$ RLWE samples in $R_q \times R_q$.  Then the problem $\SRLWE(\cR)$ can be solved using $S$ by $n/f$ calls to the oracle $\sA$, $2mn/f$ reductions $R_q \rightarrow R/\fq$, and $2mn/f$ evaluations of a Galois automorphism on $R_q$.
\end{theorem}

\begin{proof}
        The Galois group $G =\operatorname{Gal}(K/\bQ)$ acts on the set $\{\fq_1, \cdots ,\fq_g\}$ transitively. Hence for each $i$, there exists $\sigma_i \in \operatorname{Gal}(K/\bQ)$, such that $\sigma_i(\fq) = \fq_i$.  
	Next, remark that $\sigma_i^{-1}(S)$ is itself a set of RLWE samples in $R_q \times R_q$, since the action of Galois is an isometry of the Minkowski embedding.  (Here it is essential that we consider only spherical Gaussians.)	
  Furthermore, the secret for this set of samples is $\sigma_i^{-1}(s)$.  	Then we call the oracle $\sA$ on the input $(\sigma_i^{-1}(S) \pmod \fq, \fq)$.  The algorithm will output $\sigma_i^{-1}(s) \pmod{\fq}$, from which we can recover $s \pmod{\fq_i}$ using $\sigma_i$.  We do this for all $1\leq i \leq g = n/f$ and use \eqref{eqn: Rq-factor} of Lemma \ref{lem: prime factorization} to recover $s$.
 \end{proof}

In particular, if the number of samples $m$ is polynomial in $n$ and the time taken to evaluate Galois automorphisms on a single sample is also polynomial in $n$, then Theorem~\ref{thm: reduction} gives a polynomial time reduction from $\SRLWE(\cR)$ to $\SRLWE(\cR,\fq)$.

\begin{remark}
        For a proper runtime analysis of the reduction, one must examine the implementation, in particular with regards to Galois automorphisms.
The runtime for evaluating an automorphism depends
rather strongly on the instance and on the way ring elements are represented. For example, for subfields of cyclotomic fields represented with respect to normal integral bases, the Galois automorphisms are simply permutations of the coordinates, so the time needed to apply these automorphisms is trivial.
\end{remark}

%\begin{remark}
%Although the theorem is stated for any unramified prime. From an %attacker's perspective, we still take primes of small degree, since the search space for $s \pmod{\fq}$ is of size $q^f$, and it is bad when $f$ is large.
%\end{remark}

The search-to-decision reduction will follow from the lemma below.
\begin{lemma}
There is a probabilistic polynomial time reduction from $\SRLWE(\cR,\fq)$ to $\DRLWE(\cR)$.
\end{lemma}

\begin{proof}
This is a rephrasing of \cite[Lemma 5.9 and Lemma 5.12]{lyubashevsky2013ideal}.\end{proof}

\begin{corollary}\label{cor: s-to-d}
Suppose $\cR$ is an RLWE instance where $K$ is Galois and $q$ is an unramified prime in $K$. Then there is a probabilistic polynomial-time reduction from $\SRLWE(\cR)$ to $\DRLWE(\cR)$.
\end{corollary}

\section{The Chi-square Attack}
\label{sec: chi-square}

In this section, we extend the $f(1) \equiv 0 \pmod q$ attack of \cite{eisentrager2014weak} and the root-of-small-order attack of \cite{elos2015weak}.  These attacks can be viewed as examples of a more general  principle, as follows.  Suppose one has a ring homomorphism
\[
        \phi: R_q \rightarrow F
\]
where $F$ is a finite field, and where two properties hold:
\begin{enumerate}
        \item $F$ is small enough that its elements can be examined exhaustively; and
        \item the error distribution on $R_q$, transported by $\phi$ to $F$, is detectably non-uniform.
\end{enumerate}

Then the attack on decision RLWE  is as follows:
\begin{enumerate}
        \item Transport the samples $(a, b)$ in $R_q \times R_q$ to $F \times F$ via $\phi$.
        \item Loop through possible guesses for the image of the secret, $\phi(s)$, in $F$.
        \item For each guess $g$, compute the distribution of $\phi(b) - \phi(a)g$ on the available samples. Note that if we 
        let $g^* = \phi(s)$ denote the true value, 
        \[
        		\phi(b) - \phi(a) g = (\phi(b) - \phi(a) g^*)  + \phi(a)(g-g^*) = \phi(e) +  \phi(a)(g-g^*), 
        \]
        which equals $\phi(e)$ if the guess is correct, and looks uniform otherwise. 
        \item If the samples are RLWE samples with secret $s$ and $g = \phi(s)$, then this distribution will follow the error distribution, which will look non-uniform.
        \item If all such distributions look uniform, then the samples were uniform, not RLWE, samples.
\end{enumerate}

The fact that $\phi$ is a ring homomorphism is essential in guaranteeing that for the correct guess, the distribution in question is the image of the error distribution.  The only ring homomorphisms from $R_q$ to a finite field are given by reduction modulo a prime ideal $\fq$ lying above $q$ in $R$.

\subsection{Chi-square Test for Uniform Distribution}
We briefly review the properties and usage of the chi-square test for uniform distributions over a finite set $S$. We partition $S$ into $r$ subsets $S = \bigsqcup_{j=1}^r S_j$, called \emph{bins}.
Suppose there are $M$ samples $y_1, \ldots, y_M \in S$.
For each $1 \leq j \leq r$, we compute the expected number of samples in the $j$-th subset: $c_j := \frac{|S_j|M}{|S|}$. Then we compute the actual number of samples in $S_j$, i.e., $t_j := |\{1 \leq i \leq r: y_i \in S_j\}|$. Finally, the $\chi^2$ value is computed as
\[
    \chi^2(S,y) = \sum_{j = 1}^r \frac{(t_j -c_j)^2}{c_j}.
\]
Suppose the samples are drawn from the uniform distribution on $S$. Then the $\chi^2$ value follows the chi-square distribution with $(r-1)$ degrees of freedom, which we denote by $\chi_{r-1}^2$. Let $\cF_{r-1}(x)$ denote its cumulative distribution function. For the chi-square test, we choose a confidence level parameter $\alpha \in (0,1)$ and compute $\delta = \cF_{r-1}^{-1}(\alpha)$. Then we reject the hypothesis that the samples are drawn from the uniform distribution if $\chi^2(S,y)  > \delta$.

If $P,Q$ are two probability distributions on the set $S$, then their {\it statistical distance} is defined as
$d(P,Q) = \frac{1}{2} \sum_{t \in S} |P(t) - Q(t)|$. For convenience, we also define the {\it $l_2$ distance} between $P$ and $Q$ as $d_2(P,Q) = (\sum_{t \in S} |P(t) - Q(t)|^2)^{\frac{1}{2}}$. We have the inequality $d(P,Q) \leq \frac{\sqrt{|S|}}{2}d_2(P,Q)$.

\begin{remark}
We chose to use the chi-square test for our attack since we are distinguishing a known distribution (uniform on $R/\fq$) from an unkown distribution (discrete Gaussians mod $\fq$). If the latter distribution is efficiently computable, then one might switch to other statistical tests, e.g., the Neyman-Pearson test, for better results.
\end{remark}

\subsection{The Chi-square Attack on $\SRLWE(\cR,\fq)$}

Let $\cR$ be an RLWE instance with error distribution $D_{\Lambda_R,\sigma}$ and $\fq$ be a prime ideal above $q$.  The basic idea of our attack relies on the assumption that the distribution $D_{\Lambda_R,\sigma} \mod {\fq}$ is distinguishable from the uniform distribution on the finite field $F = R/\fq$. More precisely, the attack loops through all $q^f$ possible values $\bar{s} = s \pmod{\fq}$, and for each guess $s'$, it computes the values $\bar{e}' = \bar{b} - \bar{a} s' \pmod {\fq}$ for every sample $(a,b) \in S$. If the guess is wrong, or if the samples are taken from the uniform distribution in $(R_q)^2$, the values $\bar{e}'$ would be uniformly distributed in $F$ and it is likely to pass the chi-square test. On the other hand, if the guess is correct, then we expect the test on the errors $\bar{e}'$ to reject the null hypothesis. Let $N := q^f$ denote the cardinality of $F$. We remark that as a general rule of thumb for the 
chi-square test, we need to generate at least $5N$ samples. 

For the attack to be successful, we need the $(N-1)$ tests corresponding to wrong guesses of $s \pmod{\fq}$ to pass, and the one test corresponding to the correct guess to be rejected. For this purpose, we need to choose the confidence level $\alpha$ to be close enough to one (a reasonable choice is $\alpha = 1 - \frac{1}{10N}$). The detailed attack is described in Algorithm~\ref{alg: chi-square}.  Let $\cF_{N-1}(x)$ denote the cumulative distribution function of $\chi_{N-1}^2$.
%Let $\beta$ denote the probability that the sample errors fails the uniform test with probability  Then the probability that our algorithm will success is $p  = (1- \frac{1}{10N})^{N-1} \beta$. Note that when $N$ is large, $(1- \frac{1}{10N})^{N-1}$ is about $e^{-1/10} \approx 0.904$.

%Note that although we restrict ourselves to subfields of cyclotomics with odd and square-free $m$, the attack could be applied to any finite extension of $\bQ$.

\begin{algorithm}
\caption{chi-square attack on $\SRLWE(\cR,\fq$)}
 \label{alg: chi-square}        % give the algorithm a caption
              % and a label for \ref{} commands later in the document
\begin{algorithmic} % enter the algorithmic environment
    \Require  $\cR = (K,q,\sigma, s)$ -- an RLWE instance; $R$ -- the ring of integers of $K$; $\fq$ -- a prime ideal in $K$ above $q$; $F = R/\fq$ -- the residue field of $\fq$; $N = q^{f}$ -- the cardinality of $F$; $\mathcal{S}$ -- a collection of $M$ ($M = \Omega(N)$) RLWE samples from $\cR$; $0 < \alpha < 1$ -- the confidence level.
    \Ensure a guess of the value $s \pmod{\fq}$, or {\bf NOT-RLWE}, or {\bf INSUFFICIENT-SAMPLES}
    \State $\delta \gets \cF_{N-1}^{-1}(\alpha)$, $\cG \gets \emptyset$.
    \For{$s$ in $F$}
        \State $\cE \gets \emptyset$.
        \For{$a,b$ in $\mathcal{S}$}
            \State $\bar{a}, \bar{b} \gets a \pmod{\fq}, b \pmod{\fq}$.
            \State $\bar{e} \gets \bar{b} - \bar{a}s$.
            \State add $\bar{e}$ to $\cE$.
        \EndFor

        \State     $\chi^2(\cE) \gets \sum_{j = 1}^N \frac{(|\{c \in \cE: c = j\}|  - M/N)^2}{M/N}$.

        \If{$\chi^2(\cE) >  \delta$}
            \State add $s$ to $\cG$.
        \EndIf
    \EndFor
    \If{$G = \emptyset$}

        \Return {\bf NOT-RLWE}
    \ElsIf{$G = \{g\}$}

        \Return $g$
    \Else

        \Return {\bf INSUFFICIENT-SAMPLES}
    \EndIf

\end{algorithmic}
\end{algorithm}

\begin{remark}
For simplicity of exposition, we use $N$ bins in Algorithm~\ref{alg: chi-square}, that is one element per bin. In some situations,
it might be advantageous to choose the bins differently.
\end{remark}

The time complexity of the attack is $O(nq^{2f})$ since there are $q^{f}$ possible values for $s \pmod {\fq}$, each reduction modulo $\fq$ takes $O(n)$ to compute, and the number of samples needed is $O(q^f)$. The correctness of the attack is captured in Theorem~\ref{thm: attack} below. For $\lambda \in \bR$ and $d \in \bZ$, we use $\cF_{d,\lambda}(x)$ to denote the cumulative distribution function of the noncentral chi-square distribution with degree of freedom $d$ and parameter $\lambda$.

\begin{theorem} \label{thm: attack}
Let $\cR  = (K,q,s, \sigma)$ be an RLWE instance. Suppose $\fq$ be a prime ideal in $K$ above $q$, and let $\Delta$ denote the statistical distance between the distribution $D_{\Lambda_R, \sigma} \mod{\fq}$ and the uniform distribution on $R/\fq$.  Let $M$ be the number of samples used in Algorithm~\ref{alg: chi-square}, and let $\lambda = 4 M \Delta^2$. Let $0 < \alpha < 1$ and let $\delta = \cF_{N-1}^{-1}(\alpha)$. If $p$ is the probability of success of the attack in Algorithm~\ref{alg: chi-square}, then

$$p  \geq \alpha^{N-1} (1-  \cF_{N-1; \lambda}(\delta) ).$$
\end{theorem}

\begin{proof}
It is a standard fact (see \cite{ryabko2004new}, for example) that the chi-square value on samples from $D_{\Lambda_R, \sigma} \mod{\fq}$ follows the noncentral chi-square distribution with $(N-1)$ degrees of freedom and parameter $\lambda_0$ given by
\[
    \lambda_0 =  d_2(D_{\Lambda_R, \sigma} \Mod{\fq}, U(R/\fq))^2 \cdot MN.
\]
Note that we have $\lambda_0 \geq (2d(D_{\Lambda_R, \sigma} \mod{\fq}, U(R/\fq))/\sqrt{N})^2 MN = 4M\Delta^2 = \lambda$. Recall that our attack succeeds if the ``error'' set $\cE$ from each of the $(N-1)$ wrong guesses of $s \pmod{\fq}$ passes the test, and the true reduced errors fail the test. We assume that the results of these tests are independent of each other. Then the first event happens with probability $\alpha^{N-1}$, whereas the second event has probability  $1-  \cF_{N-1; \lambda_0}(\delta)$. Since this is an increasing function in $\lambda_0$, we can replace $\lambda_0$ by $\lambda$, and the theorem follows.
 \end{proof}

\begin{remark}
One could choose the value of $\alpha$ in Theorem~\ref{thm: attack} to suit the specific instance.
The probability of success will change accordingly. When we expect the statistical distance $\Delta$ to be large, it is preferable to choose a larger $\alpha$ to increase the probability of success.  For example, if we choose $\alpha = 1-\frac{1}{10N}$, then $\alpha^{N-1} \geq e^{-1/10} = 0.904 \cdots$.
\end{remark}

Figure~\ref{fig: probability} shows a plot of $p$ versus $\Delta$ for various choices of $N$, made according to Theorem~\ref{thm: attack}, where we fix the number of samples to be $M = 5N$ and fix $\alpha = 1-\frac{1}{10N}$.
\begin{figure}
\begin{center}
\includegraphics[width = 0.60\textwidth]{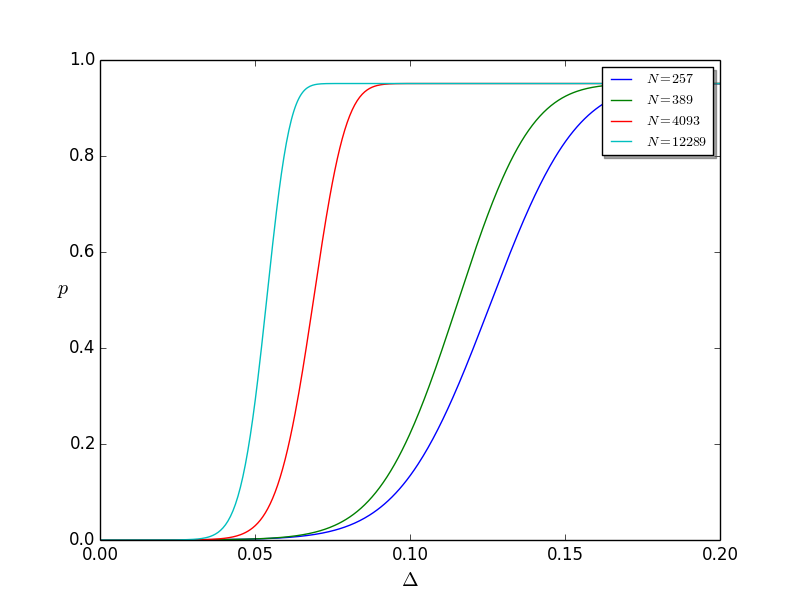}
\caption{Success probability versus statistical distance}
\label{fig: probability}
\end{center}
\end{figure}

\begin{remark}
For linear equations with small errors, there is the attack on the search RLWE problem proposed by Arora and Ge \cite{arora2011new}. However, the attack requires solving a linear system in $\approx n^d/d!$ variables. Here $d$ is the number of possible values for the error: e.g., if the error can take values $0, 1, 2, -�1, -��2$, then $d = 5$. Since it requires $\approx n^d/d!$ samples, the attack of Arora and Ge requires $\geq 10^8$ samples when $n \geq 100$ and $d \ge 5$, for example. In contrast, the complexity of our attack depends linearly on $n$ and quadratically on $q$. In particular, it does not depend on the error size (although the success rate does depend on the error size).
\end{remark}

\section{Vulnerable Instances among Subfields of Cyclotomic Fields}\label{sec: sub-cyclotomics}
We searched for instances of RLWE vulnerable to the chi-square attack.  For this purpose, we restricted attention to subfields of cyclotomic fields $\bQ(\zeta_m)$. Throughout this section, we assume $m$ is a positive integer that is {\it odd and squarefree}. The Galois group $\Gal(\bQ(\zeta_m)/\bQ)$ is canonically isomorphic to $G = (\bZ/m\bZ)^*$. For a subgroup $H$ of $G$, let $K_{m,H} = \bQ(\zeta_m)^H$ be the subfield of elements fixed by $H$.
Then the extension $K_{m,H}/\bQ$ is Galois with degree $n = \frac{\varphi(m)}{|H|}$. Also, the residue degree of a prime $q$ in $K_{m,H}$ is equal to the order of $q$ in the quotient group $G/H$. Moreover, $K_{m,H}$ has a canonical {\it normal integral basis}, as follows. For each integer $i$ coprime to $m$, set $w_{i} =  \sum_{h \in H} \zeta_m^{hi}$. Then  $w := (w_{i})_{i \in G/H}$ is a $\bZ$-basis of $R$. (For a proof of this fact, see \cite[Proposition 6.1]{johnston2011notes}). Thus we have $A_w = TA_w'$, where \[
T = \begin{bmatrix}
    I_{r_1}  & 0 & 0  \\
    0     & \frac{1}{\sqrt{2}} I_{r_2}  &  \frac{1}{\sqrt{2}} I_{r_2} \\
    0  &    \frac{-i}{\sqrt{2}} I_{r_2}     & \frac{i}{\sqrt{2}}I_{r_2}
\end{bmatrix}
, \, \quad
    (A_w')_{ij} = \sum_{h \in H}{\zeta_m^{hij}}, \mbox{ for } i, j \in G/H.
\]
\begin{remark}
The field $K_{m,H}$ is totally real if and only if $-1 \in H$, in which case $(r_1, r_2) = (n,0)$. Otherwise, it is totally
complex, and $(r_1,r_2)  =(0,n/2)$.
\end{remark}

\begin{lemma} \label{lem: symmetry}
Suppose $\cR$ is an RLWE instance such that the underlying field $K$ is a Galois number field and that $q$ is unramified in $K$. Then the reduced error distribution $D_{\Lambda_R,\sigma} \mod{\fq}$ is independent of the choice of prime ideal $\fq$ above $q$.
\end{lemma}

\begin{proof}
        From Lemma \ref{lem: prime factorization}, we may switch from a prime $\fq$ to $\fq'$ via $\Gal(K/\bQ)$. On the other hand, the Galois group acts on the lattice $\Lambda_R$ by permuting the coordinates. Hence we have a group homomorphism $$\phi: \Gal(K/\bQ) \to \Aut(\Lambda).$$ Since permutation matrices are orthogonal, the Galois group action on $\Lambda_R$ given by $\phi$ is distance-preserving. In particular, it preserves any spherical discrete Gaussian distribution on $\Lambda_R$.
\end{proof}

\subsection{Searching for Vulnerable Instances}

Algorithm~\ref{alg: chi-square} allows us to search for vulnerable instances among fields of the form $K_{m,H}$ by generating actual RLWE samples and running the attack. Success of the attack will indicate vulnerability of the instance. Note that our field searching requires sampling efficiently from a discrete Gaussian $D_{\Lambda, \sigma}$, for which we use the efficient algorithm developed in \cite{gentry2008trapdoors}.

In Table~\ref{tab: attacked}, we list some instances on which the attack has succeeded. The columns of Table~\ref{tab: attacked} are as follows. The first two columns specify $m$ and the generators of $H$, where $H$ is represented as
a subgroup of $(\bZ/m\bZ)^*$; the column labeled $f$ is the residue degree of $q$. The last column consists of either the runtime for an actual attack which succeeded, or an estimation of the runtime. Note that we omitted our choice of prime ideal $\fq$, since due to Lemma~\ref{lem: symmetry} the choice of $\fq$ is irrelevant to our attack. The parameters $\sigma_0$ in Table~\ref{tab: attacked} represent the boundary of the power of our attack, i.e., we tried higher $\sigma_0$ and the attack failed. Note that although $\sigma_0$ is relatively small, in practice it still provides exponentially many error vectors. Intuitively, when $\sigma_0$ = 1, our $\sigma$ is equal to the geometric mean of the lengths of a Gram-Schmidt basis of $\Lambda_R$. In practice, the lengths of these basis vectors do not differ by a lot, so we still expect to get at least $\Omega(2^n)$ error vectors.

Also, in terms of the normal integral basis of $R$, the coefficients of the error $e$ are of the same size. In particular, none of the coefficients will be zero with overwhelming probability. Thus a standard linear algebra attack does not apply to this case.

The rows of Table~\ref{tab: attacked} with ``estimated'' runtime mean the following. First, we ran the chi-square test on the correct reduced errors to obtain an estimate $\hat{\Delta}$ of the statistical distance $\Delta$. We then chose $\alpha$ according to $\hat{\Delta}$ and obtained an estimation $\hat{p}$ of the success probability of our attack, using the formula in Theorem~\ref{thm: attack}.  The corresponding rows in the table all have $\hat{p} > 1 - 2^{-10}$, suggesting that the attack is very likely to succeed.  Finally, we ran a few chi-square tests on samples obtained from a few randomly chosen incorrect guesses to compute the average time $t$ for running one chi-square test. We set the estimated runtime for the attack to be $tN$.

\begin{table}
\caption{Attacked sub-cyclotomic RLWE instances}
\label{tab: attacked}
\begin{center}
\begin{tabular}{c|c|c|c|c|c|c|c}
$m$ & generators of $H$ & $n$ & $q$ & $f$ & $\sigma_0$ & no. samples & runtime (in hours) \\ \hline
2805 &  [1684, 1618] & 40 & 67 & 2 & 1 & 22445 & 3.49 \\
15015 & [12286, 2003, 11936] & 60 & 43 & 2 & 1 & 11094 & 1.05 \\
15015 & [12286, 2003, 11936] & 60 & 617 & 2 & 1.25 & 8000 & 228.41 (estimated)  \\
90321 & [90320, 18514, 43405] & 80 & 67 & 2 & 1 & 26934 & 4.81 \\
255255 &  [97943, 162436, 253826, 248711, 44318] & 90 & 2003 & 2 & 1.25 & 15000 &  1114.44 (estimated) \\
285285 & [181156, 210926, 87361] & 96 & 521  & 2 & 1.1 & 5000 & 75.41 (estimated) \\
1468005 & [312016, 978671, 956572, 400366] & 100 & 683 & 2 & 1.1
& 5000 &  276.01 (estimated) \\
1468005 & [198892, 978671, 431521, 1083139] & 144 & 139 & 2 & 1 &  4000 &  5.72
\end{tabular}
\end{center}
\end{table}

%More precisely, suppose $\widehat{\chi^2}$ is the chi-square value of the sample errors from $D_{\cR, \fq}$. We replace $\lambda$ by $\widehat{\chi^2}$ in the formula and compute
%%    \widehat{p}  = 0.904 \left(1 - \Phi \left(\frac{\Phi^{-1}(1- \frac{1}{20N})\sqrt{2(N-1)}- \widehat{\chi^2}}{\sqrt{2(N-1) +4\widehat{\chi^2}}}\right)\right).
%\]
%The value $\widehat{p}$ is then our estimate of the sucess rate of our attack.  In addition, we estimate the runtime based on the average time taken for the tests we've done.

\begin{remark}
Castryk et al~\cite{castryck2016provably} show that there are even more weaknesses in the weak instances found in~\cite{elos2015weak} so that one can attack the corresponding search RLWE problem using a standard linear algebra attack using only a few samples.  The approach  from~\cite{castryck2016provably} using linear algebra will not work on the examples in this paper.  Although each coordinate of the error vector only takes on small integers, it is unlikely that any fixed coordinate of the error vector will equal to zero. Hence one can not hope to extract exact linear equations from the samples. 

Nonetheless, in~\cite{castryck2016error} Castryck et al. performed an analysis of the instances in our Table~\ref{tab: attacked}, and showed that one can recover a certain number of approximate linear equations from each RLWE sample. One can certainly run the Arora-Ge attack using these approximate equations. However, a careful analysis shows that for instances in our table, our attack is more efficient than the Arora-Ge attack. 

For example, we take the first instance from Table~\ref{tab: attacked}. In Section 5 of~\cite{castryck2016error},  it is shown that out of each RLWE sample, one can recover 20 noisy linear equations in the secret key, with each noise sampled from a Gaussian of mean zero and standard deviation 0.5381. (In \cite{castryck2016error}, the standard deviation was incorrectly claimed to be $0.5381/\sqrt{2\pi}$ due to a misunderstanding: they used $r = 1$ in their analysis but our instance has $r = \sqrt{2\pi}$.) 

First, we try $d = 7$. In order to 
run the Arora-Ge attack, we need in the best case 
$({{40+ 7-1}\choose{7}} / 20) \approx 2^{21}$ RLWE samples, assuming all errors after rounding to integers lies in [-3,3]. If we choose $d = 5$ instead, then in the best case we need $({{40+ 5-1}\choose{5}} / 20) \approx 54300$ RLWE samples. However, to achieve this, we need the rounded  errors in each equation to lie in [-2,2], which happens  with probability $\erf(\frac{2.5}{\sqrt{2}\cdot 0.5381})^{{40+ 5-1}\choose{5}} \approx 0.025$. Our attack on the other hand requires 22445 samples and succeeds with probability greater than 1/2. Moreover, the computational complexity of Arora-Ge attack is cubic in the number of samples, while our attack is linear in the number of samples. Hence we conclude that our attack is more efficient. A similar analysis can be done for other instances in Table~\ref{tab: attacked}.

% Took out reference to Peikert
%Most recently, in \cite{peikert2016not}, the author suggests that the instances in this paper are vulnerable to a distinguishing attack based on the trace pairing.  In fact, our examples \emph{do not} fit into the framework of \cite{peikert2016not}, since they involve maps to degree two finite fields, whereas \cite{peikert2016not} concerns prime fields.  As we explain further in an appendix, the attacks described in \cite{peikert2016not} \emph{do not succeed} on the examples in this paper.  On the contrary, the examples in this paper are vulnerable because of a novel property of the underlying number field, not seen in any previous attacks:  for more detail, see the discussion in Section 5.2.
\end{remark}

\subsection{Discussion of the Reason for Vulnerability} \label{subfield vulnerability}

We searched for vulnerable instances where the modulus has residue degree one or two. It turns out that all vulnerable instances we found and listed in Table~\ref{tab: attacked} have a modulus of degree two.  In this section we give a heuristic explanation for the existence of examples of higher degree.  
Let $K$  be a Galois number field and suppose $q$ is a prime of residue degree $f$ in $K$.  We will give a scenario under which a vulnerability to our attack may appear.

For the purposes of the thought experiment, we will suppose there exists a ``good'' integral basis $w_1,\cdots, w_n$ of the ring of integers $R$, by which we mean that the vectors $\iota(w_i)$ and $\iota(w_j)$ are almost orthogonal and short for $i \neq j$; this is only for convenience in the discussion. 
Fix a prime ideal $\fq$ above $q$. Then the images of the basis under the reduction modulo $\fq$ map are elements of $F := R/\fq$. Now if for some index $i$, the element $w_i$ lies inside some proper subfield $K'$ of $K$, and if $q$ has residue degree $f' < f$ in $K'$, then $w_i \pmod{\fq}$ will lie in a proper subfield of $F$. \emph{If this occurs for a large number of the basis elements $w_i$}, then we could  expect the distribution $D_{\Lambda_R, \sigma} \mod{\fq}$ to take values in a proper subfield of $F$ more frequently than the uniform distribution. This would allow us to distinguish it from the uniform distribution on $F$.

In practice, we found that the above scenario is particularly likely when the field $K$ has a subfield $K'$ of index 2 such that $q$ splits completely in $K'$ and has residue degree 2 in $K$. Since the ring of integers of $K'$ is a subring of the ring of integers of $K$, one has at least $n/2$ linearly independent vectors in $\Lambda_R$ with the desired property, i.e., their reduction modulo some prime $\fq$ above $q$ lie inside $\bF_q$ instead of $\bF_{q^2}$.

\subsection{A Detailed Example} \label{subsec: detailed}
In order to illustrate our discussion above together with the search-to-decision reduction, we present a vulnerable Galois  instance in detail, where we generated RLWE samples, performed the attack, and used the search-to-decision reduction to recover the entire secret $s$.

\begin{example}
Let $m = 3003$ and $H$ be the subgroup of $(\bZ/m\bZ)^*$ generated by 2276, 2729 and 1123. Then $K = K_{m,H}$ is a Galois number field of degree $n = 30$. We take the modulus to be $q = 131$, a prime of residue degree 2, and take $\sigma_0 = 1$. We generate the secret $s$ from the discrete Gaussian $D_{\Lambda_R, \sigma}$.  There are 15 prime ideals in $K$ lying above $q$, which we denote by $\fq_1, \cdots, \fq_{15}$. We then generate 1000 RLWE samples and use Algorithm~\ref{alg: chi-square} and Theorem~\ref{thm: reduction} to recover $s \pmod {\fq_i}$ for each $1 \leq j \leq 15$. Then we use the Chinese remainder theorem to recover $s$. The attack succeeded in 32.8 hours. The code for this attack is in the appendix.
\end{example}

\section{Attacks on the Prime Cyclotomic Fields}
\label{sec: ramified-prime}

\subsection{Attacking non-dual RLWE when $q = p$}
Let $p$ be an odd prime and let $K = \bQ(\zeta_p)$ be the $p$-th cyclotomic field. Then $K$ has degree $(p-1)$ and discriminant $p^{p-2}$. The prime $p$ is totally ramified in $K$, so there is a unique prime ideal $\fp = (1 - \zeta_p)$ above $p$, and the reduction from $R/pR$ to $R/\fp R \cong \bF_p$ takes all powers of $\zeta_p$ to 1.

We give a heuristic argument that the attack could work: writing the error $e$ as  $\sum e_i \zeta_m^i$, we have $e \pmod{\fp} = \sum_i e_i$. Since the coefficients $e_i$ tend to be small, it may be that  $e \pmod{\fp}$ takes on small values with higher probability, making the instance vulnerable to our chi-square attack. Table~\ref{tab: ramified} contains data of some actual attacks we have done. Note that the parameters $\sigma_0$ represent the boundary of the power of our attack in practice, i.e., we tried higher $\sigma_0$ and the attack failed.

\begin{table}[H]
\begin{center}
\caption{Attacked instances of DRLWE for $K= \bQ(\zeta_p)$}
\label{tab: ramified}
\begin{tabular}{c|c|c|c}
$q$ $( = p)$ & $n$ & $\sigma_0$ & runtime (in seconds) \\
\hline
251 & 250 &  0.5 & 2.62\\
503 &  503 & 0.575 & 12.02\\
809 & 808 & 0.61 & 34.38\\
\end{tabular}
\end{center}
\end{table}

\subsection{Attacking dual RLWE}
\label{attackingdual}

We adopt our attack to the decision version of dual RLWE for the field $K = \bQ(\zeta_p)$, with no assumptions on the modulus $q$. Keep the notations as above, and let $R^{\vee}$ be the dual ideal of $R$. Let $r>0$ be the width parameter. Then the error $e$ is sampled from the continuous spherical Gaussian distribution of width $r$, which is denoted $D_r$ in \cite{lyubashevsky2013ideal}. Recall that the secret $s \in R^{\vee}/qR^{\vee}$, and an RLWE sample is $(a, b= as+e) \in R_q \times K_\bR/qR^{\vee}$.

We start by scaling the second coordinate by $p$. Then $b' = bp = a(ps) + pe$. Using the fact that $pR^{\vee} = \fp$, we see that $s' = ps \in \fp/q\fp$, and $e' = pe \in K_\bR/q\fp$. Thus we can regard $s'$ as the new secret, and $e'$ as the new error.

Note that the scaled error $e' = pe$ is sampled from the continuous spherical Gaussian $D_{pr}$. Equivalently, by \cite{ducas2012ring}, we may assume $e'$ is sampled as
\[
    e' = p \cdot e = \sum_{i=0}^{p-1}  e_i \zeta_p^i.
\]
where the coefficients $e_i$ are i.i.d. one-dimensional Gaussians with width $\sqrt{p}r$.

Recall that our goal is to tell the difference between the above samples and samples chosen uniformly from $R_q \times K_\bR/q\fp$. Let $\beta = \zeta_p - 1$. Note that $K = \bQ(\beta)$, hence every element in $K_\bR$ can be uniquely written as $\sum_{i=0}^{p-2} a_i \beta^i$ ($a_i \in \bR$). Consider the map
\[
    \rho: K_\bR \to \bR: \sum a_i \beta^i \mapsto a_0. \tag{*}
\]
It is clear that $\rho$ is additive. We examine the image of $e'$ under the map $\rho$. We write
\[
	e' = \sum_{i=0}^{p-2} \epsilon_i \beta^i,  (\epsilon_i \in \bR).
\]
Then one verfies that $\epsilon_0 = e_0 + \cdots + e_{p-2} - (p-1)e_{p-1}$ and we have $\rho(e) = \epsilon_0$.

%(We note the difference between this expression and the PLWE distribution. In the PLWE distribution we have a sum of the $e_i$'s, so we have a square root cancellations and we can attack much larger parameters; here the appearance of the ``extra term'' $(p-1)e_{p-1}$ sort of limits our attack to small parameters.)

Now we make two observations: first, since the ideal $\fp$ is generated by $\beta$, we have $\rho(\fp) = p\bZ$;  second, we have $as' = a (p \cdot s) \in \fp/q\fp$. Combining these observations, we see that
$$\rho(b') = \rho(as) + \rho(e') \equiv \rho(e') \equiv \epsilon_0 \mod{p\bZ}.$$

We could describe our attack on the decision RLWE as follows: for each scaled sample $(a,b')$, we compute $\rho(b')$. Then we perform a statistical test on the set $\{\rho(b') \mod{p\bZ}\} \subseteq \bR/p\bZ$ to distinguish it from the uniform distribution on the circle $\bR/p\bZ$.

Note that this attack did not involve the modulus $q$, thus it can be applied to {\it any} modulus. This is in contrast to the previous attack on the non-dual case, where the attack was only performed under the assumption that $q = p$ is the unique ramified prime.

%\begin{remark}
%As a curious note, we observe that the above attack does not reveal anything about the secret $s$. This raise the following question: the map $\rho$ defined in (*) can be viewed as taking the first coefficient of a polynomial in $\beta$.
%One can generalize this: consider the map $\rho_k$ which takes the first $k$ coefficients. Then, using the idea of \cite{elos2015weak}, one might be able to guess some partial information on $s$, and then apply it to obtain the distributions of $\rho_k(e') \in \bR^k/\Lambda$ for some lattice $\Lambda$, and test for uniformity. This approach might allow us to attack the decision dual RLWE on prime cyclotomic fields with larger width parameters.
%\end{remark}

\begin{remark}
The search-to-decision reduction for dual RLWE in cyclotomic fields and completely split modulus is proved in
\cite{lyubashevsky2013ideal}. However, the theorem requires that the error width $r \geq \eta_\epsilon(R^{\vee})$ for some negligible $\epsilon = \epsilon(n)$. (Here $\eta_\epsilon(R^{\vee})$ is the smoothing parameter defined in
\cite{micciancio2007worst}. For $R = \bZ[\zeta_p]$, if we take $\epsilon = 2^{-p+1}$, then one sees that $\eta_\epsilon(R^{\vee}) \leq 1$, and $\eta_\epsilon(R^{\vee})$ tends to 1 in the limit as $p \to \infty$).  Hence the search-to-decision reduction of  \cite{lyubashevsky2013ideal} essentially requires $r \geq 1$, which is above the parameters we can attack ($r \sim 1/\sqrt{p}$).
So in this particular case, our attack on the decision problem cannot be transferred to an attack on the search problem using this search-to-decision reduction.
\end{remark}

Table~\ref{dual prime cyclotomic} records some successful attacks. Note that we have omitted the modulus $q$ since it
is irrelevant to the attack.  We used 50$\sim$400 bins for the chi-square tests. We observe from the table that the error width we can attack is about a constant times $1/\sqrt{p}$, and that the constant is growing (if slowly) as $p$ grows.

\begin{table}[h!]
\caption{Attacking the dual RLWE in $\bQ(\zeta_p)$}
\begin{center}
\label{dual prime cyclotomic}
\begin{tabular}{c|c|c|c|c}
$p$ & $r\sqrt{p}$ & no. samples & average run time & success rate \\ \hline % 50 bins
307 & 0.82 & 1535 & 0.048 second& 6 out of 10  \\
507 & 0.83 & 2515 & 0.076  second & 8 out of 10 \\
809 & 0.85 &  4045 & 0.134 second & 6 out of 10 \\
997 & 0.86 & 4985 & 0.154 second & 5 out of 10\\
1103 & 0.87 & 5515 & 0.192 second & 5 out of 10 \\
1201 & 0.88 & 6005 & 0.202 second& 2 out of 10 \\
\end{tabular}
\end{center}
\end{table}

\section{Can Modulus Switching be Used?}
\label{sec:mod}

The modulus switching procedure is a technique to reduce noise
in RLWE samples, and has been discussed extensively in  \cite{brakerski2012leveled} and \cite{langlois2014worst}.
We recap the basic ideas of modulus switching. Let $\cR = (K, q, \sigma, s)$ be an RLWE instance. Choose another prime $p$ less than $q$ as the new modulus and consider the instance $\cR' = (K,p,\sigma',s)$ for some $\sigma' > \sigma$.  We can ``switch modulus'' if there exists a map
\[
\pi_{q,p} : R_q \to R_p,
\]
which takes RLWE samples with respect to $\cR$ to RLWE samples with respect to $\cR'$.  In what follows, we give a heuristic argument that our attack will not work in combination with modulus switching under a na\"ive implementation, and isolate the key characteristics a successful implementation of the attack would require.

One example of a map $\pi_{q,p}$ being used in practice is as follows. Let $\alpha = \frac{p}{q}$ and fix a small positive number $\tau$. For an equivalence class $[a]$ in $R_q$, we sample a vector $a'$ from the ``shifted discrete Gaussian'' $D_{\Lambda_R, \tau, \alpha a}$,  defined as follows. For a lattice $\Lambda$ and a vector $c \in \bR^n$,
\[
    D_{\Lambda, \tau, c}(x) = \frac{\rho_\tau(x - c)}{\sum_{y \in\Lambda} \rho_\tau(y-c)}, \, \forall x \in \Lambda.
\]
Finally, we set $ \pi_{q,p}([a]) = a' \pmod{pR}$. Note that the definition of $\pi_{q,p}([a])$ is independent of the choice of  representative $a$, as follows.  Suppose we choose another representative $a_1$, then $a_1 = a+ \lambda q$ for some $\lambda \in R$, hence $\alpha a_1 = \alpha a + \lambda p$.  Finally, observe that the shifted discrete Gaussian behaves well under translating by a lattice point, i.e., we have $D_{\Lambda, \tau, c+u} = D_{\Lambda, \tau, c} + u$ for any $u \in \Lambda$.

Put loosely, the map $\pi_{q,p}$ scales $a$ by $p/q$ and then rounds back into the lattice.  It is a natural question then to ask whether modulus switching can be combined with our attack, to switch from a ``strong'' modulus to a ``weak'' modulus.  However, a heuristic argument shows that the naive combination of our attack with modulus switching will not work.

Let $a'' = \alpha a - a'$. By construction, we expect $a''$ to be a  short vector in $\bR^n$, and the point $a'$ can be viewed as a ``rounding'' of the point $\alpha a$ to the lattice $\Lambda_R$.

We will make two heuristic assumptions:

\begin{enumerate}
        \item That $\pi_{q,p}$ takes the uniform distribution on $R_q$ to an almost uniform distribution on $R_p$.
        \item The distribution of $b''$ and $(sa)''$ is independent modulo $\fq$, for $s \neq \pm 1$.
\end{enumerate}

\begin{proposition}
        Under the assumption that $\pi_{q,p}$ takes the uniform distribution on $R_q$ to an almost uniform distribution on $R_p$, the reduction of $a''$ modulo $\mathfrak{p}$ will be almost uniformly distributed in $R/\mathfrak{p}R$.
\end{proposition}

\begin{proof}
The reduction map $R \rightarrow R/\mathfrak{p}$
is a ring homomorphism that can be extended to a homomorphism of additive groups $\phi: \frac{1}{q} R \rightarrow R/\mathfrak{p}$ by the following chain of maps:
\[
        \frac{1}{q} R \xrightarrow{\pmod{\mathfrak{p}\frac{1}{q}R}} \frac{1}{q}R \Big/ \mathfrak{p}\frac{1}{q}R \xrightarrow{ \times q} R/\mathfrak{p}R \xrightarrow{ \times [q]^{-1}} R/\mathfrak{p}R.
\]
Then the relation $a'' + a' = \alpha a$ is preserved by this map.  However, $\phi(\alpha a) = 0 \pmod{\mathfrak{p}}$, so that $\phi(a'') \equiv -\phi(a')$.
\end{proof}

Suppose we have a sample $(a,b)$ and the switched sample $(a', b') = (\pi_{q,p}(a),\pi_{q,p}(b))$. Consider the error $e':= b' - a's$. Suppose $b = as+e+ \lambda q$ for some $\lambda \in R$. Then
\begin{align*}
    e' &= b' - a's  \\
    &= \alpha(b-as) - b''  + a''s. \\
    & = \alpha e + \lambda p - b'' + a''s.
\end{align*}
and therefore, considering this as an additive relation in $\frac{1}{q}R$ and applying the map of the proof above,
$$e' \equiv - b'' + a'' s \pmod{\fp}.$$

By the Proposition above, $a''$ and $b''$ are uniformly distributed modulo $\fp$.  Hence, if we assume the $a''$ and $b''$ are independent, then the reduced rounding errors $a'' \pmod{\fp}$ and $b'' \pmod{\fp}$ are also independent, and the new reduced errors $e' \pmod{\fp}$ would follow the uniform distribution. So our chi-square attack will fail on these modulus-switched samples, even though $p$ might be a ``weak'' modulus.

Therefore, the best hope of attack is if one of our two assumptions is violated by a map $\pi_{q,p}$.  The second is the most likely target.  Note that $a''$ and $b''$ are the rounding errors when we try to round $\alpha a$ and $\alpha b$ to the lattice $\Lambda_R$. However, $\Lambda_R$ is a $n$-dimensional lattice, so there are $\Omega(2^n)$ options of rounding a vector in $\bR^n$ to a moderately close lattice point. Even in the scenario with zero error, i.e., $e = 0$, an attacker will face the task of finding a ``nice'' rounding algorithm, so that the roundings of the two vectors $\alpha a$ and $\alpha b = \alpha a s$ are somehow related.

So far, we are not aware of any such algorithm, unless the secret $s$ is trivial, e.g., $s = 1$, in which case $\alpha a$ is almost equal to $\alpha b$, and one expects that $a''$ is close to $b''$.

\bibliographystyle{splncs}
\bibliography{galois-rlwe}

\section{Appendix A: Code}

\subsection{SubgroupModm.sage}

This file contains the object needed for manipulating subgroups $H$ of $(\bZ/m\bZ)^*$.

\scriptsize
\begin{verbatim}

class SubgroupModm:
    """
    a subgroup of (Z/mZ)^*
    """

    def __init__(self,m, gens, elements = None):
        self.m = m
        self.phim = euler_phi(m)
        self.Zm = Integers(m)

        newgens = []
        for a in gens:
            a = self.Zm(a)
            if not a.is_unit():
                raise ValueError('the generator %s must be a unit in the ambient group.'%a)
            newgens.append(a)

        self.gens = newgens

        if elements is None:
            print 'computing group elements...'
            t = cputime()
            self.H1 = self.compute_elements()
            print 'Time = %s'%cputime(t)
            sys.stdout.flush()
        else:
            self.H1 = elements

        self.order = len(self.H1)
        print 'group order = %s'%self.order
        sys.stdout.flush()

        self._degree = ZZ(self.phim // self.order)

        print 'computing coset representatives...'
        t = cputime()
        self.cosets = self.cosets()
        print 'Time = %s'%cputime(t)
        sys.stdout.flush()

        self._is_totally_real = self.is_totally_real()

        if not self._is_totally_real:
            merged_cosets = []
            for c in self.cosets:
                if not any([-c/d in self.H1 for d in merged_cosets]):
                    merged_cosets.append(ZZ(c))
            newcosets = merged_cosets + [-a for a in merged_cosets]
            self.cosets = newcosets

    def __repr__(self):
        return "subgroup of (Z/%sZ)^* of order %s generated by %s"%(self.m, self.order, self.gens)

    def is_totally_real(self):
        """
        The fixed field Q(zeta_m)^H is totally real if and only if -1 mod m \in H.
        """
        return self.Zm(-1) in  self.compute_elements()

    def compute_elements(self):
        """
        core function. Gives all the group elements
        """
        gens = self.gens
        result = [self.Zm(1)]
        for gen in gens:
            if gens != self.Zm(1):
                order = gen.multiplicative_order()
                pows = [gen**j for j in range(order)]
                result = set([a*b for a in result for b in pows])
        return result

    def cosets(self):
        """
        another core function, assuming we have elements, this shouldn't be hard.
        """
        Zm = self.Zm
        elts = self.H1
        m = self.m
        result =[]
        explored = []

        for a in range(m):
            if gcd(a,m) == 1 and a not in explored:
                for h in elts:
                    explored.append(h*a)
                result.append(Zm(a))
            if euler_phi(m) ==  len(result)*len(elts): # already have enough cosets
                return result

    @cached_method
    def coset(self, a):
        """
        elt -- an integer
        returns the coset representative for this element
        """
        Zm = self.Zm
        for bb in self.cosets:
            if Zm(a)/Zm(bb) in set(self.H1):
                return bb
        raise ValueError('did not find a coset.')

    def extension_degree(self,vec):
        """
        vec -- a vector indexed by cosets of self, representing an element z in K.
        return the degree of the extension QQ(z)/QQ.
        """
        try:
            vec = list(vec)
        except:
            raise ValueError('input can not be turned into a list. Please debug.')
        C = self.cosets
        ele_dict = dict([(a,b) for a,b in zip(C,vec) if b != 0])
        fixGpLen = 0
        for ll in C:
            fixed = True
            for a in ele_dict.keys():
                lla = self.coset(ll*a)
                try:
                    coef = ele_dict[lla]
                except:
                    fixed = False
                    break
                if coef != ele_dict[a]:
                    fixed = False
                    break
            if fixed:
                fixGpLen += 1
        return self._degree // fixGpLen


    def _check_cosets(self):
        """
        sanity check that the cosets has been computed correctly.
        """
        H1 = self.H1
        cosets = self.cosets
        from itertools import combinations
        return not any([c[1]*c[0]**(-1) in H1 for c in combinations(cosets, 2)])


    def __hash__(self):
        return hash((self.m,tuple(self.gens)))


    def _associated_characters(self):
        """
        Definition: a Dirichlet character chi of modulus m is associated to
        a subgroup H <= Z/mZ)^* if chi|_H = 1.

        return all the associated characters of self.
        """
        m, Zm = self.m, self.Zm
        G = DirichletGroup(m)
        H1 = Set(self.compute_elements())

        result =[]
        for chi in G:
            ker_chi = Set([Zm(a) for a in chi.kernel()]) # a list of integers
            if H1.issubset(ker_chi):
                result.append(chi)
        return result

    def multiplicative_order(self, a):
        """
        return the multiplicative order of [a] in the quotien group G/H
        """
        m = self.m
        Zm = self.Zm
        if gcd(m,a) != 1:
            raise ValueError
        a = Zm(a)
        o = self._degree
        for dd in o.divisors()[:-1]:
            if a**dd in self.H1:
                return dd
        return o

    def discriminant(self):
        """
        return, up to sign, the discriminant of the fixed field of self as a subfield of Q(zeta_m).
        """
        return prod([chi.conductor() for chi in self._associated_characters()])

    def intersection(self, other):
        """
        intersection of two subgroups of the same m.
        """
        if self.m != other.m:
            raise ValueError('the underlying m of self and other must be same.')
        H1 = self.H1
        H1other = other.H1
        Hnew = Set(H1).intersection(Set(H1other))
        print 'size of intersection = %s'%len(Hnew)
        Hnew_reduced = _reduce_gens(self.m,Hnew)
        print 'reduced gens for intersection  = %s'%Hnew_reduced
        sys.stdout.flush()
        return SubgroupModm(self.m, Hnew_reduced, elements = Hnew)

def _reduce_gens(m,H1):
    """
    given a full group, get a short list of generators.
    """
    Zm = Integers(m)
    gens = set([])
    gensSpan = set([Zm(1)])
    for a in H1:
        if Zm(a) not in gensSpan:
            sys.stdout.flush()
            ordera = Zm(a).multiplicative_order()
            alst  = [Zm(a)**j for j in range(1, ordera)]
            newelts = set([cc*aa for cc in gensSpan for aa in alst])
            gensSpan  |=  newelts
            gens.add(a)
        if len(gensSpan) == len(H1):
            # found enough generators.
            return list(gens)
    raise ValueError('did not find enough generators.')
\end{verbatim}

\subsection{MyLatticeSampler.sage}

\normalsize
This file allows sampling from discrete lattice Gaussian distributions using the algorithm in \cite{gentry2008trapdoors}.
It took the current implementation in sage and modified it slightly to fix some issues. The authors claim no originality of any code in this file.
\scriptsize
\begin{verbatim}
from sage.stats.distributions.discrete_gaussian_integer import DiscreteGaussianDistributionIntegerSampler
def _fpbkz(A, K = 10**20, block = 8, delta = 0.75):
    """
    including a transpose operation.
    """
    print 'blocksize for bkz = %s'%block
    At = A.transpose()
    RF = A[0][0].parent()
    AA = Matrix(ZZ, [[ZZ(round(K*a)) for a in row] for row in list(At)])
    F = FP_LLL(AA)
    F.BKZ(block_size = block, delta= delta)
    B = F._sage_()
    T = B*AA**(-1)
    B1 = Matrix(RF, [[a/RF(K) for a in row] for row in list(B)])
    return T.transpose().change_ring(ZZ), B1.transpose()


class MyLatticeSampler:
    """
    Sampling from discrete Gaussian.
    """

    def __init__(self,A,sigma = 1,dps = 60, method = 'LLL', block = None, already_orthogonal = False, gram_schmidt_norms = None):
        self.A = A # we are using column span instead of rowspan
        self.sigma = sigma

        print 'reducing the lattice...'
        t = cputime()
        self._degree = A.nrows()
        if method == 'LLL':
            self.T =  self._lll_reduce()
        elif method == 'BKZ':
            self.T = self._bkz_reduce(block = block)
        else:
            print 'no reduction is done.'
            self.T = identity_matrix(self._degree)
        self.B  = self.A*self.T
        print 'reduction done. Time: %s'%cputime(t)


        print 'Gram Schmidting...'
        t = cputime()

        if already_orthogonal: # The columns of A are already gram-schmidt.
            self._G = self.A
            if gram_schmidt_norms is None:
                self._gs_norms = [self._G.column(i).norm() for i in range(self._degree)]
            else:
                self._gs_norms = gram_schmidt_norms
        else:
            # Compute the gram-schmidt ourselves. Can be slow.
            self._gs_norms, self._G = self.compute_G(dps =  dps)
        print 'Gram Schmidt done. Time: %s'%cputime(t)

        self.final_sigma = sigma*(prod(self._gs_norms))**(1/self._degree)


    def _bkz_reduce(self,block = None):
        print 'bkz being performed...'
        if block is None:
            block = min(50, ZZ(self._degree // 2))
        return _fpbkz(self.A, block = block)[0]

    def _lll_reduce(self):
        print 'lll being performed...'
        A = self.A
        return gp(A).qflll().sage()

    @cached_method
    def col_sum(self):
        """
        related to the evaluation attack, return the list a where
                   a[i] = colsum(A^-1,i)
        """
        return vector([1 for _ in range(self._degree)])*(self.A**(-1))

    def babai_quality(self):
        """
        inspired by Kim's explanation, I think the quality of a basis
        for babai should be the ratio ||\tilde{bn}||/||\tilde{b1}||
        """
        gs_norms = self._gs_norms
        return float(min(gs_norms)/max(gs_norms))

    def __repr__(self):
        return 'Discrete Gaussian sampler with dimension %s and sigma = %s'%(self._degree, self.final_sigma.n())

    def compute_G(self, dps = 50):
        t = cputime()
        B = self.B
        n = self._degree
        from mpmath import *
        mp.dps = dps
        prec = dps*6
        AA = mp.matrix([list(w) for w in list(B)])
        Q,R = qr(AA) # QR decomposition

        M = mp.matrix([list(Q.column(i)*R[i,i]) for i in range(n)]);
        M_sage = Matrix([[RealField(prec)(M[i,j]) for i in range(n)] for j in range(n)])
        verbose('gram schmidt computation took %s'%cputime(t))
        return [abs(RealField(prec)(R[i,i])) for i in range(n)], M_sage

    def set_sigma(self,newsigma):
        self.final_sigma = newsigma

    def babai(self,c):
        """
        run babai's algorithm and find a lattice vector close to the
        input point c.
        Note this is super similar to the __call__ function

        Returns a tuple (v,z), where v is the actual vector in R^n,
        and z is its coordinate *in terms of a*. So we have
        v = Az.
        """
        n = self._degree
        try:
            c = vector(c)
        except:
            pass
        G, norms  = self._G, self._gs_norms
        B = self.B
        T = self.T
        zs = []
        v = c

        for i in range(n)[::-1]:
            b_ = G.column(i)
            v_ = v.dot_product(b_) / norms[i]**2
            z = ZZ(round(v_))
            v = v - z*B.column(i)
            zs.append(z)
        return c - v, T*(vector(zs[::-1]))

    def __call__(self, c = None):
        """
        c -- an n-dimensional vector, so that we are sampling a discrete gaussian
        centered at c.
        """
        v = 0
        sigma, G = self.final_sigma, self._G
        n = self._degree
        if c is None:
            c = zero_vector(n)
        B = self.B
        T = self.T
        zs = []
        norms = self._gs_norms
        for i in range(n)[::-1]:
            b_ = G.column(i)
            c_ = c.dot_product(b_) / norms[i]**2
            sigma_ = sigma/norms[i]
            assert(sigma_ > 0)
            z = DiscreteGaussianDistributionIntegerSampler(sigma=sigma_, c=c_, algorithm="uniform+table")()
            c = c - z*B.column(i)
            v = v + z*B.column(i)
            zs.append(z)
        return v, T*vector(zs[::-1])
\end{verbatim}

\subsection{SubCycSampler.sage}

\normalsize This file allows generating the errors and reducing them modulo prime ideals when the field $K$ is a subfield of some cyclotomic field with
odd and squarefree $m$.

\scriptsize
\begin{verbatim}
from sage.stats.distributions.discrete_gaussian_integer import DiscreteGaussianDistributionIntegerSampler
import sys

class SubCycSampler:
    """
    We write our own GPV sampler for sub-cyclotomic fields.
    It also has the functionality of simulating an attack.

    Caution: according to GPV, we need to have s >= ||tilde(B)||*log(n)
    for the sampler to approximate discrete lattice Gaussian. So if
    s is smaller than what is required, the __call__() method is not
    guaranteed to output discrete Gaussian.
    """

    def __init__(self,m,H,sigma = 1,prec = 100, method = 'BKZ',block = None):
        """
        require: m must be square free and odd.

        disc: the discriminant of K =  Q(zeta_m)^H. We pass it
        as an optional parameter, since when the order of H is
        large, the computation could be very slow.
        """

        self.m = m
        self.H = H

        self.H1 =self.H.H1
        sys.stdout.flush()


        self.cosets = H.cosets

        self.sigma = sigma
        self.prec = prec


        t = cputime()
        self._degree = euler_phi(m) // len(self.H1)

        self._is_totally_real = self.H._is_totally_real


        print 'computing embedding matrix...'
        t = cputime()
        self.TstarA, self.Acan = self.embedding_matrix(prec = self.prec)
        self.Acaninv  = None
        print 'time = %s'%cputime(t)
        sys.stdout.flush()

        self.D = MyLatticeSampler(self.TstarA, sigma = self.sigma, method = method, block = block)
        self.Ared = self.D.B

        self._T = self.D.T

        self.final_sigma  =self.D.final_sigma
        self.secret = self.__call__()


    def __repr__(self):
        return 'RLWE error sampler with m = %s,  H = %s, secret  = %s and sigma = %s'%(self.m, self.H, self.secret, self.final_sigma.n())

    def minpoly(self):
        K.<z> = CyclotomicField(self.m)
        return sum([z**h for h in self.H1]).minpoly()

    def compute_G(self, prec = 53):
        """
        computing a colum gram-schmidt basis for the embedded lattice O_K.
        return the basis and the length of each vector as a list.

        Modified on 8/2: do this after using LLL to reduce the basis.
        """
        B = self.Ared
        n = self._degree
        from mpmath import *
        mp.dps = prec // 2
        BB = mp.matrix([list(ww) for ww in list(B)])
        Q,R = qr(BB) # QR decomposition
        M = mp.matrix([list(Q.column(i)*R[i,i]) for i in range(n)]);
        M_sage = Matrix([[RealField(prec)(M[i,j]) for i in range(n)] for j in range(n)])
        v = [abs(R[i,i]) for i in range(n)]
        return M_sage,v # vectors are columns

    def degree_of_prime(self,q):
        """
        return the degree of q in K
        """
        if not q.is_prime():
            raise ValueError('q must be prime')
        return (self.H).multiplicative_order(q)

    def degree_n_primes(self, min_prime, max_prime, n =1):
        """
        return a bunch of primes of degree n in K. When n = 1, this
        is split primes.
        """
        result = []
        for p in primes(min_prime, max_prime):
            try:
                if self.degree_of_prime(p) == n:
                    result.append(p)
            except:
                pass
        return result

    def basis_lengths(self):
        return [self.Ared.column(i).norm() for i in range(self._degree)]


    def galois_permutation(self, c):
        """
        c -- a coset.
        returns a dictionary d such that  d[a] = \sigma_c(a),
        representing a Galois group action.
        """
        H = self.H
        Zm = Integers(self.m)
        c = Zm(c)
        d = {}
        for a in self.cosets:
            d[a] = self.H.coset(a*c)
        return d

    def _vec_modq_coset_dict(self,q):
        vec = self.vec_modq(q)
        cc = self.cosets
        return dict(zip(cc,vec))

    def vec_modq_twisted_by_galois(self,q,c, reduced = False):
        _dict = self._vec_modq_coset_dict(q)
        _galois = self.galois_permutation(c)
        result  = []
        for a in self.cosets:
            result.append(_dict[_galois[a]])
        if not reduced:
            return vector(result)
        else:
            return vector(result)*self._T

    def embedding_matrix(self, prec = None):
        """
        We are in a simplified situation because the field K is Galois over QQ,
        so it is either totally real or totally complex.
        to-do: can optimize this.
        """
        m = self.m
        H1 = self.H1
        if prec is None:
            prec = self.prec
        C = ComplexField(prec)
        zetam = C.zeta(m)
        cosets = self.cosets
        n = self._degree

        _dict = {}
        for l in cosets:
            _dict[l]  = sum([zetam**(ZZ(l*h)) for h in H1])

        A = Matrix([[_dict[self.H.coset(l*k)] for l in cosets] for k in cosets])

        if self._is_totally_real:
            Areal = _real_part(A)
            return Areal, A
        else:
            T = t_matrix(n,prec = prec)
            return _real_part(T.conjugate_transpose()*A),A

    def coset_reps(self):
        """
        I need this for representing the basis vectors. Each coset rep c
        represents the element \alpha_c =  \sum_{h \in H} \zeta_m^{ch}.
        """
        return self.cosets

    def __call__(self,c = None):
        """
        return an integer vector a = (a_c) indexed by the coset reps of self,
        which represents the vector \sum_c a_c \alpha_c
        Use the algorithm of [GPV].
        http://www.cc.gatech.edu/~cpeikert/pubs/trap_lattice.pdf

        If minkowski = True, return the lattice vector in R^n. Otherwise,
        return the coordinate of the vector in terms of the embedding matrix of self.
        """
        return self.D(c = c)[1]

    def babai(self,c):
        return self.D.babai(c)[1]

    def _modq_dict(self,q):
        """
        a sanity check of the generators modulo q.
        """
        cc = self.cosets
        vv = self.vec_modq(q)
        return dict(zip(cc,vv))

    def subfield_quality(self):
        """
        portion of elements of our reduced basis that lie in proper subfields.
        """
        T = self._T
        count = 0
        for i in range(S._degree):
            col  = T.column(i)
            sys.stdout.flush()
            deg = S.H.extension_degree(col)
            print 'degree of Q(b_i) = %s'%deg
            sys.stdout.flush()
            if deg < S._degree:
                count += 1
        return float(count/S._degree)



    def subfield_quality_modq(self,q, twist = None):
        if twist is None:
            vq = self.vec_modq(q,reduced = True)
        else:
            vq = self.vec_modq_twisted_by_galois(q,twist, reduced = True)
        F = vq[0].parent()
        deg = F.degree()
        return float(len([aa for aa in vq if aa.minpoly().degree() < deg])/self._degree)


    @cached_method
    def vec_modq(self,q, reduced = False):
        """
        the basis elements (normal integral basis) modulo q.

        If reduced is true, return the LLL-reduced basis mod q

        v dot Tz = (vT) dot z
        """
        m = self.m
        degree = self.degree_of_prime(q)
        v = finite_cyclo_traces(m,q,self.cosets,self.H1, deg = degree) # could be slow
        if not reduced:
            result = vector(v)
        else:
            result =  vector(v)*self._T
        return result

    def _to_ccn(self, lst):
        """
        convert an element in O_K from C^n to Z^n.
        """
        return list(self.Acan*vector(lst))


    def _to_zzn(self,lst):
        """
        the inversion of the above.
        """
        if self.Acaninv is None:
            self.Acaninv = (self.Acan)**(-1)
        return list(self.Acaninv*vector(lst))


    def _prod(self,lsta, lstb):
        """
        multiplying two field elements using the canonical embedding
        """
        lsta, lstb = list(lsta), list(lstb)
        lsta_cc, lstb_cc = self._to_ccn(lsta), self._to_ccn(lstb)
        float_result = self._to_zzn([aa*bb for aa, bb in zip(lsta_cc,lstb_cc)])
        return [ZZ(round(tt.real_part())) for tt in float_result]

    def set_sigma(self,newsigma):
        self.D.final_sigma = newsigma

    def set_secret(self, newsecret):
        self.secret = newsecret
\end{verbatim}

\subsection{Chisquare.sage}

\normalsize
This file implements a variant of the chi-square test over finite fields $F_{q^f}$ for $q$ a prime and $f > 1$.

\scriptsize
\begin{verbatim}
def subfield_unifrom_test(samples, probThreshold = 1e-5):
    """
    Assume that the samples are from a finite field.
    we separate the ones that are from a proper subfield.
    """
    F = samples[0].parent()
    q = F.characteristic()
    degF = F.degree()
    numsamples = len(samples)
    eltsWithFullDegree = elts_of_full_degree(q,degF)
    nSmall = 0
    nLarge = 0
    for aa in samples:
        if aa.minpoly().degree() < degF:
            nSmall +=1
        else:
            nLarge +=1
    card = q**degF
    eLarge= float(eltsWithFullDegree/card*numsamples)
    eSmall= numsamples - eLarge
    verbose('eSmall, eLarge = %s,%s'%(eSmall, eLarge))
    verbose('nSmall, nLarge = %s,%s'%(nSmall, nLarge))
    if min(eSmall, eLarge) < 5:
        raise ValueError('samples size too small.')

    chisquare = (nSmall - eSmall )^2/eSmall + (nLarge - eLarge)^2/eLarge
    T = RealDistribution('chisquared', 1)
    verbose('chisquare = %s'%chisquare)
    prob = 1 - T.cum_distribution_function(chisquare)
    if prob < probThreshold:
        verbose('non-uniform')
        return False
    else:
        verbose('uniform')
        return True
\end{verbatim}

\subsection{Example of an attack}
\label{subsec: code}

\normalsize
This code implements the attack on an Galois RLWE instance described in Section~\ref{subsec: detailed}.

\scriptsize
\begin{verbatim}
print 'We peform the full attack on an Galois instance.'

totaltime = cputime()
import sys

load('SubgroupModm.sage','MyLatticeSampler.sage','SubCycSampler.sage','Chisquare.sage')

def _my_dot_product(lst1,lst2):
    return sum([a*b for a,b in zip(lst1,lst2)])

m = 3003; H = SubgroupModm(m, [2276, 2729, 1123]);
S = SubCycSampler(m,H);

sigma0 = 1.0

S = SubCycSampler(m,H,prec = 300, method = 'LLL', sigma = sigma0)

print 'S = %s'%S

q = 131
degq = H.multiplicative_order(q)
print 'degree of prime q = %s is %s'%(q, degq)
sys.stdout.flush()

print 'final sigma = %s'%S.final_sigma

print 'degree of field = %s'%(euler_phi(m)//H.order)
sys.stdout.flush()

numsamples = 1000;
print 'generating %s errors...'%numsamples
sys.stdout.flush()
errors = []
for dd in range(numsamples):
    error = S()
    errors.append(error)
    if dd > 0 and Mod(dd,1000) == 0:
        print '%s/%s samples generated'%(dd, numsamples)
        print 'an example error is %s'%error
        sys.stdout.flush()
print 'error generation done.'
sys.stdout.flush()
save(errors, 'errors.sobj')

vq = S.vec_modq(q)
print 'vq = %s'%vq
sys.stdout.flush()
F = vq[0].parent()
sys.stdout.flush()

Flst = [a for a in F]
alpha = F.gen()
Fp = F.prime_subfield()
print 'defining polynomial of F = %s'%alpha.minpoly()
sys.stdout.flush()

print 'Generating uniform a...'

alst = [[ZZ.random_element(q) for _  in range(S._degree)] for jj in range(numsamples)]
print 'Generation of uniform a done.'
sys.stdout.flush()

s = [ZZ.random_element(q) for _  in range(S._degree)]
print 'secret = %s'%s
sys.stdout.flush()

# The attack.
success, SUCCESS = True, True
count = 1
for cc in S.cosets:
    t = cputime()
    success = True
    print 'coset %s/%s with representative %s'%(count, S._degree, cc)
    sys.stdout.flush()

    count += 1
    vqcc = S.vec_modq_twisted_by_galois(q,cc)
    smodq = F(_my_dot_product(s,vqcc))
    print 'smodq = %s'%smodq
    sys.stdout.flush()

    amodqlst,bmodqlst = [],[]
    for a,e in zip(alst, errors):
        emodq = F(_my_dot_product(e,vqcc))
        amodq = F(_my_dot_product(a,vqcc))
        amodqlst.append(amodq)
        bmodqlst.append(amodq*smodq+emodq)

    countsmall = 0
    for sguess in Flst:
        countsmall +=1
        sys.stdout.flush()
        if Mod(count, 1000) == 0 or sguess == smodq:
            print 'example run: %s/%s runs'%(count,len(Flst))
            print 'sguess = %s'%sguess
            if sguess == smodq:
                print 'this is the correct guess'
            set_verbose(1)
        reducedErrors = [bb - aa*sguess for aa, bb in zip(amodqlst, bmodqlst)]
        uniform = subfield_unifrom_test(reducedErrors,  probThreshold = 1e-10)
        if uniform and sguess == smodq:
            print 'failed to detect'
            success = False
            break
        elif (not uniform) and sguess != smodq:
            print 'uniform is distorted'
            success = False
            break
        set_verbose(0)
    print 'Done computing with coset [%s]. success = %s'%(cc, success)
    sys.stdout.flush()

    print 'Time taken = %s'%cputime(t)
    SUCCESS = SUCCESS and success
    sys.stdout.flush()
    print '*'*20

print '#'*40
print 'Summary:'
print 'H = %s'%H
print 'degree of field = %s'%(euler_phi(m)//H.order)
print 'q = %s, degree of q =  %s'%(q, degq)
print 'sigma_0 = %s'%sigma0
print 'number of samples = %s'%numsamples
print 'success? : %s'%SUCCESS
print 'Total Time = %s'%cputime(totaltime)
sys.stdout.flush()
\end{verbatim}

\end{document}